\documentclass[envcountsame]{llncs}

\usepackage{microtype}
\usepackage{hyperref}
\usepackage{latexsym}
\usepackage{amsmath,amssymb,stmaryrd}
\usepackage{graphicx}
\usepackage{breakurl}\hypersetup{colorlinks=false}
\usepackage{versions}\excludeversion{ignore}
\newcommand{\Continuous}[1]{\mathit{Flow}(#1)}   
\newcommand{\notContinuous}[1]{\mathit{Jump}(#1)}   
\newcommand{\Machine}{\textsc{Dynamic}}
\newcommand{\St}{\ensuremath{\mathcal {S}}}

\newcommand{\R}{\ensuremath{\mathbb{R}}}

\newcommand{\N}{\ensuremath{\mathbb{N}}}
\newcommand{\I}{\ensuremath{\mathbb{I}}}

\newcommand{\TT}{\ensuremath{\mathbb{T}}}
\newcommand{\In}{\St_0}

\newcommand{\U}{\ensuremath{\mathcal {V}}}

\newcommand{\evaluation}[2][]{\ensuremath{\left\llbracket #2\right\rrbracket_{#1}}}
\newcommand{\val}[2]{\evaluation[#2]{#1}}
\newcommand{\vall}[2]{\val{#2}{#1}}

\newcommand{\True}{\textsf{true}}

\newcommand{\update}[3]{(#1,#2,#3)}
\spnewtheorem{postulate}{Postulate}{\bfseries}{\itshape}
\spnewtheorem{postulatep}{Postulate}{\bfseries}{\itshape}

\usepackage{color}

\renewcommand\eqref[1]{(\ref{#1})}

\newcommand\LOCATION[2]{(#1,#2)}

\newcommand{\J}[1]{\{\JJ\}\times #1}
\newcommand{\C}[1]{\{\CC\}\times #1}
\newcommand{\JJ}{\mathcal{J}}
\newcommand{\CC}{\mathcal{F}}
\newcommand{\gen}[1]{\Delta_t(#1)}

\newcommand\Deltapsi{\Delta_\psi}
\newcommand{\itm}[1]{\mbox{\rm(}#1\mbox{\rm)}}
\newcommand{\nd}[1]{#1}
\let\ep\endproof
\renewcommand{\endproof}{\qed\ep}
\newenvironment{ttcode}{\begin{ttfamily}\rm\tt}{\end{ttfamily}}
\newcommand\s{\phantom{x}}
\begin{document}
\title{Axiomatizing Analog Algorithms}
\author{Olivier Bournez\inst{1}\thanks{This author's research was
    partially supported by a French National Research
    Agency's grant (ANR-15-CE40-0016-02).}
  \and Nachum
  Dershowitz\inst{2}\thanks{This author's research benefited from
    a fellowship at the Institut d'\'Etudes Avanc\'ees de Paris
    (France), with the financial support of the French National Research Agency's ``Investissements
    d'avenir'' program (ANR-11-LABX-0027-01 Labex RFIEA+).}
\and Pierre N\'eron\inst{3}}
\institute{Laboratoire d'Informatique de l'X (LIX), \'Ecole
  Polytechnique, France \and
School of Computer Science, Tel Aviv University, Ramat Aviv, Israel
\and French Network and Information Security Agency (ANSSI), France \\
\email{bournez@lix.polytechnique.fr}~~\email{nachum@cs.tau.ac.il}~~\email{pierre.neron@ssi.gouv.fr}}
\authorrunning{O. Bournez and N. Dershowitz and P. N\'eron}

\maketitle

\begin{abstract}
  We propose a formalization of generic algorithms that includes analog
  algorithms.  This is achieved by reformulating and extending the framework of abstract
  state machines to include continuous-time models of computation.
     {We prove that every hybrid algorithm satisfying some reasonable postulates   may be expressed precisely by
   a program in a simple and expressive language.}
\end{abstract}

\section{Introduction}

In \cite{Gur00}, Gurevich showed
that any algorithm that satisfies three intuitive ``Sequential
Postulates'' can be step-by-step emulated by an {abstract state machine (ASM)}.  These postulates
formalize the following intuitions: (I) one is dealing with
discrete deterministic state-transition systems; (II) the information in states
suffices to determine future transitions and may be captured by
logical structures that respect isomorphisms; and (III) transitions
are governed by the values of a finite and input-independent set of
ground
terms.
All notions of algorithms for ``classical'' \emph{discrete-time} models of
computation in computer science are covered by this formalization.
This includes Turing machines, random-access memory (RAM) machines,
and their sundry extensions.
The geometric constructions in \cite{Reisig}, for example, are
loop-free examples of discrete-step continuous-space (real-number) algorithms.
The ASM formalization also covers general discrete-time models evolving over continuous
space like the Blum-Shub-Smale machine model \cite{BSS}.

However, capturing \emph{continuous-time} models of computation is
still a challenge, that is to say, capturing models of computation that operate
in continuous (real) time and with real values.
Examples of continuous-time models of computations include models of analog
machines like the General Purpose Analog Computer (GPAC) of Claude Shannon
\cite{Sha41}, proposed as a mathematical model of the Differential
Analyzers, built for the first time in 1931 \cite{Bush31}, and
used to solve various problems ranging from
ballistics to aircraft design -- before the era of the digital computer
\cite{Nyc96}. 
Others include
Pascal's 1642 \textit{Pascaline},
Hermann's 1814 \textit{Planimeter},
as well as
Bill Phillips' 1949 water-run \textit{Financephalograph}.
Continuous-time computational models also include neural networks
and systems built using electronic
analog devices.
Such systems begin in some initial state and evolve over time;
results are read off from the evolving
state and/or from a terminal state. More generally, determining which
systems can actually be considered to be computational models is an
intriguing question and relates to philosophical discussions about
what constitutes a programmable machine.
Continuous-time
computation theory is far less understood than its discrete-time
counterpart \cite{Survey}.
Another line of development of continuous-time  models was
 motivated by hybrid systems, particularly by questions related to
the hardness of their verification and control.  In hybrid
 systems, the dynamics change in response to changing
conditions, so there are discrete transitions as well as continuous
ones.
Here, models are not
seen as necessarily modeling analog machines, but,
rather, as abstractions of systems about
which one would like to establish properties or derive
verification algorithms \cite{Survey}.
Some work on ASM models dealing with
continuous-time systems has been accomplished
for specific cases \cite{Cohen,Cohen2}.
Rust \cite{rust2000hybrid} specifies  forms of continuous-time evolution
based on ASMs using infinitesimals.
However, we find that a comprehensive framework
capturing general analog systems is still wanting.

Our goal is to capture all such analog and hybrid models within one uniform notion of
computation and of algorithm.
To this end, we formalize a generic notion of
continuous-time algorithm. The proposed
framework is an extension of \cite{Gur00}, as discrete-time
algorithms are a simple special case of analog algorithms.
(The initial attempt \cite{BournezDF12} was not fully satisfactory, as no
completeness theorem nor general-form result was obtained. Here, we
indeed achieve both.)
We
provide postulates defining continuous-time algorithms, in the spirit of
  those of \cite{Gur00},
  and we
prove some completeness results.
We define a simple notion of an analog ASM
program and  prove that all models satisfying the postulates
have
  corresponding analog programs (Lemma
  \ref{mlemma} and Theorem \ref{thone}).
Furthermore, we provide  conditions guaranteeing that said program is unique up
to equivalence (Theorem \ref{thtwo} and Corollary
\ref{coromain}).
All of this seamlessly extends the
results of \cite{Gur00} to analog and hybrid systems.
The proposed framework  covers all classes of
continuous-time systems that can be modeled by ordinary differential
equations or have hybrid dynamics, including the models in \cite{Survey} and
 the examples in \cite{BournezDF12}. 
 It is a first step towards a general understanding of
computability theory for continuous-time models, taken in the hope that it will also lead to a
formalization of a ``Church-Turing thesis'' for analog systems in the spirit of what has been achieved for discrete-time models~\cite{BD,dershowitz2008natural,3P}.
Systems with continuous input signals
and other means of specifying continuous behavior are left for future
work.

Some of our ideas were inspired by the way the semantics
of hybrid systems are given in the approach of
Platzer~\cite{Platzer08}.
Among attempts at studying the semantics of analog systems  within a general
framework  is
\cite{TZ}. Recent results on comparing analog models include
\cite{fu2015models}. Soundness and (relative) completeness results for
a programming language with infinitesimals have also been obtained in
\cite{suenaga2011programming}. Applications to verification have  been
explored \cite{hasuo2012exercises}.

\section{General Algorithms}

We want to generalize the notion of algorithms introduced by Gurevich in \cite{Gur00} in
order to capture not only the sequential case but also continuous
behavior. (For lack of place, we assume some familiarity with
\cite{Gur00}.)
However, when evolving continuously, an algorithm can no longer be viewed as a discrete sequence of
states, and we need a notion of evolution that can capture both kinds of behavior. This
is based on a notion of a \emph{timeline}  that corresponds
to  algorithm execution.


\begin{definition}[Time]
  \emph{Time} $\TT$  corresponds
    to a totally ordered monoid: there is an associative binary
    operation $+$, with some neutral element $0$, and a total relation $\le$ preserved
  by $+$:  $t \le t'$ implies $t+t'' \le t'+t''$ for all $t''\in\TT$.
\end{definition}

An element of $\TT$ will be called a \emph{moment}. Examples of time \TT\@
are $\R^{\geq 0}$ and $\N$. As expected, $t < t'$ will mean $t
\le t'$ but not $t = t'$.


\begin{definition}[Timeline]
A timeline is a subset of $\TT$ containing $0$. We let $\I$ denote the set of all timelines.
\end{definition}

For a moment $i \in I$ of timeline $I$,  we write $\notContinuous{i}$ if there exists
$t \in I$ with $i<t$, and there is no $t' \in I$ with $i < t' < t$.
We write $\Continuous{i}$ otherwise: that means that for all $t$, $i<t$,
there is some in-between $t' \in I$ with $i < t' < t$.
\nd{A moment $i$ with $\notContinuous{i}$ is meant to indicate a
  discrete transition.  In this case, we write $i^+$ for the smallest
  $t$ greater than $i$.}
A timeline $I$ is non-Zeno if for any moment $i \in I$, there is a
finite number of moments $j \le i$ with $\notContinuous{j}$. $\I$ is
non-Zeno if all its timelines are.

For timelines $\I=\R^{\geq 0}$, for instance, we have $\Continuous{i}$ for all $i \in
    \I$.
    For $\I=\N$, we have $\notContinuous{i}$ for all $i \in \I$, and
    $i^+=i+1$.
    We intend (for hybrid systems, in particular) to also
    consider timelines  mixing both properties, that is, with
    $\Continuous{i}$ for some $i$ and $\notContinuous{i}$ for
    other $i$. Formally building  such timelines is easy (for
      example $\bigcup_{n \in \N} [n,n+0.5]$). All these examples are non-Zeno.

\begin{definition}[Truncation]
Given a timeline $I\in\I$ and a moment $i$ of $I$, the
\emph{truncated} timeline $I[i]$ is the timeline defined by
$I[i] = \{t \mid  i+t \in I\}$.
\end{definition}



With timelines in hand, we can define hybrid dynamical systems.

\begin{definition}[Dynamical System]\label{LTS} \label{def:ts}
A \emph{dynamical  system} $\langle \St, \In, \iota, \varphi 
\rangle$
consists of the following:
\itm{a} a nonempty set (or class)
$\St$ of \emph{states};
\itm{b} a nonempty subset (or subclass) $\In \subseteq \St$, called \emph{initial} states;
\itm{c} a \emph{timeline} map $\iota : \St \to \I$, with $\I$ non-Zeno;
\itm{d} a \emph{trajectory} map $\varphi: (X : \St) \times \iota(X) \to \St$.
We require that,
for any state $X$ and moments $i, i+i' \in \iota(X)$, one has
$$\varphi(X,0) = X\,, \hspace*{8mm}
\iota(\varphi(X,i)) = \iota(X)[i]\,, \hspace*{8mm}
\varphi(X,i+i') = \varphi(\varphi(X,i),i')\,.
$$
\end{definition}

Together, the \emph{timeline} and \emph{trajectory} maps associate to each state its future
evolution. For a state $X$, $\iota(X)$ defines the timeline corresponding to the system behavior
starting from $X$, and $\varphi(X)$ defines its concrete evolution by associating to each moment in
$\iota(X)$ its corresponding state. The third condition ensures that evolution during $i+i'$ is similar to
first evolving during $i$ and then during $i'$; the preceding condition ensures a similar
property for timelines (and ensures  consistency of the last condition).

\begin{postulatep} An algorithm is a dynamical system.
\end{postulatep}


A vocabulary $\U$ is a finite collection of fixed-arity (possibly
nullary)  function
symbols,
 some functions of which may be tagged \emph{relational}.  A term whose
 outermost function symbol is relational is termed \emph{Boolean}. We
 assume that $\U$ contains the scalar (nullary) function \textit{true}.
A (first-order) \emph{structure} $X$ of vocabulary $\U$ is a nonempty set $S$, the
\emph{base set (domain)} of $X$, together with interpretations of the function
symbols in $\U$ over $S$: A $j$-ary function  symbol $f$ is interpreted as a function,
denoted  $\val{f}{X}$, from $S^j$ to $S$.  Elements of $S$ are also called elements of
$X$, or \emph{values}. Similarly, the interpretation of a term $f(t_1,\dots ,t_n)$ in $X$
is recursively defined by $\val{f(t_1,\dots ,t_n)}{X} = \val{f}{X}(\val{t_1}{X},\dots ,\val{t_n}{X})$.

Let $X$ and $Y$ be structures of the same vocabulary $\U$. An
\emph{isomorphism} from $X$ onto $Y$ is a one-to-one function $\zeta$
from the base set of $X$ onto the base set of $Y$ such that $f(\zeta
x_1,\dots,\zeta x_j)= \zeta x_0$ in $Y$ whenever $f(x_1,\dots,x_j)=x_0$
in $X$.

\begin{definition}[Abstract Transition System]\label{ATS}
An \emph{abstract transition system} is a dynamical system
whose states $\St$ are (first-order) structures over some finite vocabulary $\U$,
such that the following hold:
\begin{enumerate}
\item States are closed under isomorphism, so if $X\in \St$ is a state of the system, then any structure $Y$ isomorphic to $X$ is also a state in $\St$, and $Y$ is an initial state if $X$ is.
\item Transformations preserve the base set: that is, for every state $X \in \St$, for any $i \in \iota(X)$, $\varphi(X,i)$ has the same base set as $X$.


\item Transformations respect isomorphisms:   if $X\cong_\zeta Y$
  is an isomorphism of states $X,Y\in \St$,  then
   $\iota(X)=\iota(Y)$
   and for all $i \in \iota(X)$, $X_i \cong_\zeta Y_i$,
    where $X_i = \varphi(X,i)$, and $Y_i=\varphi(Y,i)$.

\end{enumerate}
\end{definition}

\begin{postulatep} \label{postulateats} An
  algorithm is an abstract
  transition system.
\end{postulatep}

When $\iota(X)$ is $\N$ (or order-isomorphic to $\N$) for all $X$, this corresponds precisely to the concepts
 introduced by \cite{Gur00}, considering that $\varphi(X,n) =
 \tau^{[n]}(X)$.


It is convenient to think of a structure $X$ as a memory of some kind: If
$f$ is a $j$-ary function symbol in vocabulary $\U$, and
$\overline{a}$ is a $j$-tuple of elements of the base set of $X$, then
the pair $\LOCATION{f}{\overline{a}}$ is called a
\emph{location}. We denote by $\vall{X}{f(\overline{a})}$ its interpretation
in $X$, i.e.\@ $\val{f}{X}(\overline{a})$.
If $\LOCATION{f}{\overline{a}}$ is a
location of $X$ and $b$ is an element of $X$ then
$\update{f}{\overline{a}}{b}$ is an \emph{update} of $X$.
When $Y$ and $X$ are structures over the same domain and vocabulary,
$Y\setminus X$ denotes the set of updates $\Delta^+=
  \{ \update{f}{\overline{a}}{\vall{Y}{f(\overline{a})}}  \mid  
\vall{Y}{f(\overline{a})} \neq  \vall{X}{f(\overline{a})}
\}.$

We want instantaneous evolution to be describable by updates:

\begin{definition} \label{defmachin}
An \emph{infinitesimal generator}   is \itm{a} a function $\Delta$ that  maps states $X$ to a set $\Delta(X)$ of updates,
and \itm{b} preserves
  isomorphisms: if $X\cong_\zeta Y$ is an isomorphism of states
  $X,Y\in \St$, then for all updates
  $\update{f}{\overline{a}}{b} \in \Delta(X)$, we have an isomorphic update
  $\update{f}{\overline{\zeta a}}{\zeta b} \in \Delta(Y)$.
\end{definition}

\nd{We write $\notContinuous{X}$ and say that $X$ is a \emph{jump} when $\notContinuous{0}$  in
timeline
$\iota(X)$; otherwise, we write $\Continuous{X}$ and say that it is a \emph{flow}.}
For states $X$ with
$\notContinuous{X}$, the following is natural:

\begin{definition} \label{defupdate}
The \emph{update generator} is the infinitesimal generator defined on
jump states $X$ as
$\Delta(X)=\Delta^+(X)$, where $\Delta^+(X)$ stands for
$\varphi(X,0^+) \setminus X$.
\end{definition}

To deal with flow states, we will also define some corresponding infinitesimal
generator $\Deltapsi$. Before doing so, let's see  how to go from semantics to
generators.


An \emph{initial evolution} over $S$ is a function 
whose
domain of definition is a timeline and whose range is $S$. 
An initial evolution is said to be \emph{initially constant} if it has
a constant prefix: that is to say, there is some $0<t$ such that
the function is constant over $[0\mathbin{..}t]$.

\begin{definition}[Semantics]
A \emph{semantics} $\psi$ over a class $\mathcal{C}$ of sets $S$
is a partial function mapping
initial evolutions over some $S \in \mathcal{C}$
to an element of $S$.
\end{definition}

\begin{remark} \label{rq:ex}
When $\TT=\R^{\geq 0}$, an example of
semantics over the class of sets $S$ containing $\R$ is the derivative $\psi_{\textrm{der}}$, mapping a function
$f$ to its
derivative at $0$ when that exists. When $\TT=\N$, an example of
semantics over the class of all sets would be the function $\psi_\N$ mapping
$f$ to $f(1)$. More generally, when $0 \in \TT$ is such that $\notContinuous{0}$, an example of
semantics over the class of all sets is the function $\psi_\N$ mapping
$f$ to $f(0^+)$.
\end{remark}

Consider a semantics $\psi$ over a class of sets $S$.
Let $X$ be a state whose domain is in the class and a location $\LOCATION{f}{\overline{a}}$
of $X$.  
%
Denote by $Evolution({X,\LOCATION{f}{\overline{a}}})$ the corresponding initial
evolution: that is to say, the
function given formally by $Evolution({X,\LOCATION{f}{\overline{a}})}: t \mapsto
\val{f(\overline{a})}{\varphi(X,t)}$ for $0 \le t \le I_1, t \in
\iota(X)$, for some $I_1 \in  \iota(X)$, with $I_1=0^+$ for a jump.
%
%
%
We use $\psi[X,f,\overline{a}]$ to denote the image of this evolution under $\psi$ (when it exists).

\begin{definition}[Infinitesimal generator associated with
  $\psi$] The infinitesimal generator associated with $\psi$,
 maps each state $X$, such that $\psi[X,f,\overline{a}]$ is defined for all locations, to
the set: $\Deltapsi(X) =\{ \update{f}{\overline{a}}{\psi[X,f,\overline{a}]} \mid  \LOCATION{f}{\overline{a}} \mbox{ is a location of
  $X$} ,~
Evolution({X,\LOCATION{f}{\overline{a}}}) \text{ is not initially constant} \}.$
\end{definition}

The update generator $\Delta^+$ (see Definition
\ref{defupdate}) is the infinitesimal generator associated with the
semantics $\psi_\N$ (of Remark \ref{rq:ex}) over flow states. 

From now on, we assume that some semantics $\psi$ is fixed to deal
with flow states. It could be
$\psi_{\textrm{der}}$, but it could also be another one (for example:
talking about integrals or built using infinitesimals as in \cite{rust2000hybrid}).  We denote by
$\Deltapsi$ the associated infinitesimal generator.

We are actually discussing algorithms relative to some
    $\psi$, and to be more precise, we should be refering to $\psi$-algorithms.
    The point is that not every infinitesimal generator is appropriate
    and that appropriateness is actually relative to a time domain and
    to the class of allowed dynamics over this time domain.
    \nd{To see this, keep in mind that -- when $\Deltapsi$ corresponds to derivative
    -- to be able to talk about
    derivatives, one implicitly restricts oneself to dynamics that are
    differentiable, hence non-arbitrary. In other words, one is restricting
    to a particular class of possible dynamics, and not all dynamics
    are allowed. Restricting to other classes of dynamics (for example, analytic
    ones) may lead to  different notions of algorithm.}


From the update generator $\Delta^+$ and $\Deltapsi$, we build a
generator also tagging states by the fact that they correspond to a
jump or a flow:

\begin{definition}[Generator of a State] \label{def:behaviord}
We define the \emph{tagged generator} of a state $X$, denoted $\gen{X}$, as a
function that maps state $X$ to $\{\CC\} \times \Deltapsi(X)$ when
$\Continuous{X}$ and $\Deltapsi(X)$ is defined and to $\{\JJ\} \times \Delta^+(X)$ when
$\notContinuous{X}$.
\end{definition}



Let $T$ be a set of ground terms.
We say that states $X$ and $Y$ \emph{coincide} over $T$,
if $\val{s}{X}=\val{s}{Y}$ for all $s\in T$.
This will be abbreviated $X =_{T} Y $.
%
The fact that  $X$ and $Y$ {coincide} over $T$ implies that $X$ and $Y$
necessarily share  some common elements in their
respective base sets,  at least all the $\val{s}{X}$ for $s \in T$.

An algorithm should have a finite imperative description. Intuitively, the evolution
of an algorithm from a given state is only determined by inspecting part of this
state by means of the terms appearing in the algorithm
description. The following
corresponds to the \emph{Bounded Exploration} postulate in
\cite{Gur00}.

\begin{postulatep} \label{postulatetrois}

For any algorithm, there exists a finite set $T$ of ground terms over
  vocabulary $\U$ such that for all states $X$ and $Y$ that coincide
  for $T$,  $\gen{X}$ and $\gen{Y}$ both exist and $\gen{X} = \gen{Y}$.
\end{postulatep}

A ground term of $T$ is a \emph{critical term} and a \emph{critical
  element} is the value (interpretation) of a critical term.


\begin{definition}[Analog Algorithm] \label{defalgo}
An \emph{algorithm} is an object satisfying Postulates I through III.
\end{definition}

\section{Characterization Theorem}


We now go on to define the rules of our programs (adding to those of ASM
programs  in \cite{Gur00}).

\begin{definition}
\begin{itemize}
\item \textbf{Update Rule:}  An \emph{update rule} of vocabulary $\U$ has the
  form $f(t_1,\dots,t_j):=t_0$ where $f$ is a $j$-ary function symbol
  in $\U$ and $t_1,\dots,t_j$ are ground terms over $\U$.


\item \textbf{Parallel Update Rule: }
If $R_1,\dots,R_k$ are update rules of vocabulary $\U$, then
\begin{center}
\begin{minipage}{8cm}
\begin{ttcode}
par \\
\s $R_1$ \\
\s   $R_2$ \\ 
\s $\dots$ \\ 
\s  $R_k$ \\
endpar \\ 
\end{ttcode}
\end{minipage}
\end{center}
is a \emph{parallel update rule} of vocabulary $\U$.
\end{itemize}
\end{definition}

$\gen{R_i,X}$ denotes the interpretation of a rule $R$ in state $X$
and is defined as expected: If $R$ is an update rule $f (t_1,\dots,t_j) := t_0$ then
$\gen{R,X} = \J{\update{f}{{(\val{t_i}{X},\dots,\val{t_j}{X})}}{\val{t_0}{X}}}$
and when $R$ is $\text{\tt par }R_1,\dots,R_k\ \text{\tt endpar}$ then $\gen{R,X} = \J{(d_1 \cup \dots \cup d_k)}$ where $\gen{R_i,X} = \J{d_i }$ for all $i$.

Next, we introduce rules to deal with $\mathit{Flows}$.

\begin{definition}
\begin{itemize}
\item \textbf{Basic Continuous Rule:}
  A \emph{basic continuous rule} of vocabulary $\U$ has the
  form $\Machine(f (t_1,\dots,t_j) , t_0) $ where $f$ is a symbol of arity  $j$ and $t_0,t_1,\dots,t_j$ are ground terms of vocabulary $\U$.


\item \textbf{Flow Rule:}
If $R_1,\dots,R_k$ are basic continuous rules of vocabulary $\U$, then
  \begin{center}
  \begin{minipage}{8cm}
\begin{ttcode}
flow \\ 
\s $R_1$ \\ 
\s  $R_2$ \\
\s  $\dots$ \\ 
\s $R_k$ \\ 
endflow \\ 
\end{ttcode}
  \end{minipage}
   \end{center}
is a \emph{flow rule} of vocabulary $\U$.
\end{itemize}
\end{definition}

Their semantics are then defined as follows.
If $R$ is a  basic continuous rule $\Machine(f (t_1,\dots,t_j) ,t_0)$ then
$\gen{R,X} = \C{\{\update{f}{{(a_1,\dots,a_j)}}{a_0}\}}$
where each $a_i= \val{t_i}{X}$.
If $R$ is a  flow rule with constituents $R_1,\dots,R_k$, then
$\gen{R,X} = \C{(d_1 \cup \dots \cup d_k)}$
where $\gen{R_i,X} = \C{d_i }$.

Finally, we allow conditionals:

\begin{definition}
\begin{itemize}
\item \textbf{Selection Rule:}
If $\varphi$ is a ground boolean term over vocabulary $\U$ and
  $R_1$ and $R_2$ are  rules of vocabulary $\U$ then:
   \begin{center}
   \begin{minipage}{8cm}
    \begin{ttcode}
      if $\varphi$ then \\
     \s $R_1$ \\ 
      else\\
      \s $R_2$ \\ 
      endif \\
    \end{ttcode}
   \end{minipage}
   \end{center}
is a rule of vocabulary $\U$.
\end{itemize}
\end{definition}


Given such a rule $R$ and a state $X$, if $\varphi$ evaluates to
 \textit{true} (the interpretation of scalar function \textit{true}) in  $X$ then
$\gen{R,X} = \gen{R_1,X}$ else $\gen{R,X} = \gen{R_2,X}$.

An ASM program of vocabulary $\U$ is a rule of vocabulary $\U$.
The first key result is the following, which can be seen as a
completeness result.

\begin{theorem}[Completeness] \label{mlemma} \nd{For every algorithm of vocabulary $\U$, there
  is an ASM program $\Pi$ over  $\U$ with the identical behavior:
 $\gen{\Pi,X} = \gen{X}$  for all states $X$.}
 \end{theorem}

\section{Proof of Theorem \ref{mlemma}}

Before turning to the proof of our main theorem, we reformulate and
extend several of the constructions in \cite{Gur00}.


\begin{lemma}[{\cite[Lemma 6.2]{Gur00}}] \label{lem1}
Consider an algorithm $A$, consider a state $X$ of $A$ and assume
$\notContinuous{X}$. By definition,  $\gen{X}=\J{\Delta^+(X)}$.

Consider $\update{f}{a_1,\dots,a_j}{a_0}$,
  an update of $\Delta^+(X)$.  Then all elements $a_0,a_1,\dots,a_j$
  are critical elements of $X$, that is, they correspond to values
  (interpretations) of critical terms.
\end{lemma}

\begin{proof}
The proof proceeds by contradiction. Assume that some $a_k$ does not correspond to the
value of any critical term.  One can easily consider a structure $Y$ isomorphic to $X$ which is obtained from $X$ by
replacing $a_k$ with a fresh element $b$.
By Postulate
\ref{postulateats}, $Y$ is a state and $\val{t}{Y}=\val{t}{X}$ for every
 critical term $t$.
By Postulate \ref{postulatetrois}, we know that
$\notContinuous{Y}$, and we must have:
$\gen{Y}=\J{\Delta^+(Y)}=\J{\Delta^+(X)}$.
By Postulate \ref{postulateats},
$a_k$ does not occur in  (the base set of)  $\varphi(Y,0^+)$ either. Hence, it cannot occur in
$\Delta^+(Y)= \varphi(Y,0^+) \setminus Y$. This gives the desired
contradiction.
\end{proof}

\begin{lemma}[Generalization of {\cite[Lemma 6.2]{Gur00}}] \label{lem2}
Consider an algorithm $A$ and assume $\Continuous{X}$. Then 
by definition
$\gen{X}=\C{\Deltapsi(X)}$.

Consider $\update{f}{a_1,\dots,a_j}{a_0}$,
an element of $\Deltapsi(X)$.  Then all elements $a_0,a_1,\dots,a_j$
are critical elements of $X$, that is, they correspond  to values of critical terms.
\end{lemma}

\begin{proof}
The proof proceeds by contradiction. Assume that some $a_k$ does not correspond to the
value of any critical term.  One can easily consider a structure $Y$ isomorphic to $X$ which is obtained from $X$ by
replacing $a_k$ with a fresh element $b$.

By Postulate
\ref{postulateats}, $Y$ is a state. Observe that $\val{t}{Y}=\val{t}{X}$ for every
 critical term $t$.

By Postulate \ref{postulatetrois}, we know that
$\Continuous{Y}$, and we must have:
$$\gen{Y}=\C{\Deltapsi(Y)}=\C{\Deltapsi(X)}.$$
By Postulate \ref{postulateats},
$a_k$ does not occur in  (the base set of)  $Y$. Hence it cannot occur in
$\Deltapsi(Y)$, since by Definition \ref{defmachin} elements in $\Deltapsi(Y)$ are
elements of the base set of $Y$. This gives the desired
contradiction.
\end{proof}


The following follows directly from Lemmas \ref{lem1} and \ref{lem2}.

\begin{corollary}[Corollary 6.6 of \cite{Gur00}]
For every state $X$, there exists a
rule $R^X$ such that $\gen{X}=\gen{R^X,X}$.

\end{corollary}

We now generalize some of the other lemmas from \cite{Gur00} to apply to our
more general setting.

\begin{lemma}[Generalization of {\cite[Lemma 6.7]{Gur00}}]
\label{sixsept}
If states $X$ and $Y$ coincide over the set $T$ of critical terms,
then:
$$\gen{R^X,Y} = \gen{Y}.$$
\end{lemma}

\begin{proof}
We have $\gen{R^X,Y}=\gen{R^X,X}=\gen{X}=\gen{Y}$.
The first equality holds because $R^X$ involves only critical terms
and because critical terms have the same values in $X$ and $Y$. The
second equality holds by the definition of $R^X$ (that is to say, this
is the previous corollary). The third equality
holds because of the choice of $T$ and because $X$ and $Y$ coincide
over $T$.
\end{proof}

\begin{lemma}[Generalization of {\cite[ Lemma 6.8]{Gur00}}] \label{prevle}
Suppose that $X,Y$ are states and that $\gen{R^X,Z}=\gen{Z}$ for some
state $Z$ isomorphic to $Y$ then:
$$\gen{R^X,Y}=\gen{Y}.$$
\end{lemma}

\begin{proof}
Let $\zeta$ be an isomorphism from $Y$ onto an appropriate $Z$. Extend $\zeta$ to tuples,
locations, updates and set of updates. It is easy to check that $\zeta(\gen{R^X,Y}) = \gen{R^X,Z}$.
By the choice of $Z$, $\gen{R^X,Z} = \gen{A,Z}$.

By Definition \ref{defmachin}, generators preserve isomorphisms,
thus $\gen{A, Z} = \zeta(\gen{A, Y })$ and then $\zeta(\gen{R^X , Y}) = \zeta(\gen{A, Y})$.
It remains to apply $\zeta^{-1}$ to both sides of the last equality.
\end{proof}

At each state $X$, the equality relation between critical elements
induces an equivalence relation
$$E_X(t_1,t_2) \mbox{ iff } \val{t_1}{X}=\val{t_2}{X}$$
over critical terms.

States $X$ and $Y$ are $T$-similar if $E_X=E_Y$.

\begin{lemma}[Generalization of {\cite[Lemma 6.9]{Gur00}}]\label{lemma:lem62gur}
Let $X$ be a state. Then,  for every
state $Y$ that is  $T$-similar to $X$, we have:
$$\gen{R^X,Y}=\gen{Y}.$$
\end{lemma}

\begin{proof}
Replace every element of $Y$ that
belongs to $X$ with a fresh element. This gives a structure $Z_1$ that
is isomorphic to $Y$ and disjoint from $X$. By Postulate
\ref{postulateats}, $Z_1$ is a state. Since $Z_1$ is isomorphic to $Y$,
it is $T$-similar to $Y$ and therefore $T$-similar to $X$.

Let $Z_2$ be the structure isomorphic to $Z_1$ that is obtained
from $Z_1$ by replacing $\val{t}{Y}$ with $\val{t}{X}$ for all critical
term $t$ (the definition of $Z_2$ is coherent because $X$ and $Z_1$ are
$T$-similar). By Lemma \ref{sixsept}, we have $\gen{R^X,Z_2}=\gen{Z_2}$.

Since $Z_2$ is isomorphic to $Z_1$ isomorphic to $Y$, then $Z_2$ is isomorphic to $Y$
and by Lemma \ref{prevle}, we conclude $\gen{R^X,Y}=\gen{Y}$.
\end{proof}

By previous Lemma \ref{lemma:lem62gur}, for every state $X$, there exists a boolean term $\varphi^X$ that
evaluates to $\True$ in a structure $Y$ if and only if $Y$ is
$T$-similar to $X$. Indeed, the desired term asserts that the equality
relation on the critical terms is exactly the equivalence relation
$E_X$.

Since there are only finitely many critical terms, there are
only finitely many possible equivalence relations $E_X$. Hence there
exists a finite set $\{X_1,\dots,X_m,Y_1,\dots,Y_n\}$ of states such that every
state is $T$-similar to one of the state $X_i$ or $Y_i$, and such that
 $\notContinuous{X_i}$ and $\Continuous{Y_i}$ for all $i$ (recall that the
property of being $\Continuous$ is preserved by $T$-similarity
from the previous lemma).  States  $\{X_1,\dots,X_m,Y_1,\dots,Y_n\}$ can
be chosen mutually exclusive, that is to say in different equivalence
relations. Boolean terms $(\varphi^{X_i})_i$ and $(\varphi^{Y_i})_i$  then realize a partition of the
set of states.

We can then go to the proof of Theorem \ref{mlemma}: 
Let $X_1,\dots,X_m, Y_1, \dots, Y_n$ be as above. 
The desired $\Pi$ is
\begin{center}
  \begin{minipage}{5cm}
\begin{ttfamily}
\noindent
   if $\varphi^{X_1}$ then\\  \s $R^{X_1}$\\  else \\
   \s if $\varphi^{X_2}$ then\\ \s\s  $R^{X_2}$ \\ \s else \\
\s\s\s\s  $\dots$ \\
\s\s\s\s    if $\varphi^{X_m}$ then \\
\s\s\s\s \s $R^{X_m}$ \\
\s\s\s\s  else \\
\s\s\s\s \s\s
   if $\varphi^{Y_1}$ then\\
\s\s\s\s \s\s\s $R^{Y_1}$ \\
\s\s\s\s  \s\s else \\
\s\s\s\s \s\s\s    if $\varphi^{Y_2}$ then \\
\s\s\s\s \s\s\s\s $R^{Y_2} $\\
\s\s\s\s \s\s\s  else  \\
\s\s\s\s \s\s\s\s\s\s\s  $\dots$ \\
\s\s\s\s \s\s\s\s \s\s\s   if $\varphi^{Y_{n-1}}$ then\\
\s\s\s\s \s\s\s\s \s\s\s\s  $R^{Y_{n-1}}$
\s\s\s\s \s\s\s\s \s\s\s else
\s\s\s\s \s\s\s\s \s\s\s\s  $R^{Y_{n}}$
\s\s\s\s \s\s\s\s \s\s\s endif \\
\s\s\s\s \s\s\s endif \\
\s\s\s\s  \s\s endif\\
\s\s\s\s endif\\
\s endif\\
endif
\end{ttfamily}
\end{minipage}
\end{center}
where the $R^{X_i}$ are (possibly parallel) update  rules, and the $R^{Y_i}$ are flow
rules.

\section{Extended Statements}

We are now very close to formulating our other theorems. First we define an
abstract state machine relative to semantics $\psi$.

\begin{definition} A $\psi$-abstract state machine $B$ comprises the following:
\itm{a} an ASM program $\Pi$;
\itm{b} a set $\St$ of (first-order) structures over some finite
    vocabulary $\U$ closed under isomorphisms, and a subset $\In \subseteq \St$ closed under isomorphisms;
\itm{c} a map $\iota$ and a map $\varphi$ such that  $\langle \St, \In, \iota, \varphi 
\rangle$ is an algorithm, where $\Deltapsi$ is fixed to be the infinitesimal
generator associated with $\psi$, and
for all
states $X$ in $\St$,
$\gen{\Pi,X} = \gen{X}$.
\end{definition}

By definition, a  $\psi$-abstract state machine $B$ satisfies all the
postulates and hence is an algorithm.

\begin{definition}
An ASM program $\Pi$ is \emph{$\psi$-solvable} for a set $\St$ of (first-order) structures over some finite
    vocabulary $\U$ closed under isomorphisms and a subset $\In
    \subseteq \St$ closed under isomorphisms if there exists a unique
    $\iota$ and $\varphi$ such that $(\Pi,\St,\In,\iota,\varphi)$ is a $\psi$-abstract
    state machine.
\end{definition}

\begin{definition}
A semantics $\psi$ is \emph{unambiguous} if for all sets $\St$ of (first-order) structures over some finite
    vocabulary $\U$ closed under isomorphisms, and for all subsets $\In
    \subseteq \St$ closed under isomorphisms, whenever there exists some $\iota$
    and $\varphi$ such that $(\Pi,\St,\In,\iota,\varphi)$ is a $\psi$-abstract
    state machine, then $\iota$ and $\varphi$ are unique.
\end{definition}

Our main results follow.

\begin{theorem} \label{thone}
  For every $\psi$-definable 
  algorithm $A$, there exists an equivalent $\psi$-abstract state machine $B$.
\end{theorem}

\begin{proof}
  By construction, $A$ is a hybrid dynamical
  system such that $\gen{A,X}=\gen{\Pi,X}$ for some $\Pi$ given by
  previous discussions. Set the states of $B$ to be the
  states of $A$ and the initial states of
  $B$ to the initial states of $A$.
\end{proof}

\begin{theorem} \label{thtwo}
  Assume that $\psi$ is unambiguous.
  For every $\psi$-definable 
  algorithm $A$, there exists a unique equivalent $\psi$-abstract
  state machine $B$ with same states and initial states.
\end{theorem}

\begin{proof}[of Theorem \ref{thtwo}]
This is exactly the same proof as for Theorem \ref{thone}. Unicity comes from the definition of
unambiguity.
\end{proof}

\begin{corollary} \label{coromain}
Assume that $\psi$ is unambiguous.
For every $\psi$-definable 
  algorithm $A$, there exists an equivalent $\psi$-solvable ASM
  program.
\end{corollary}

 To any 
    algorithm $A$ that is $\psi$-definable there corresponds an
    equivalent $\psi$-abstract state machine $B$, and hence a
    $\psi$-solvable program $\Pi$. Conversely, a $\psi$-abstract state
    machine $B$ corresponds to a $\psi$-definable algorithm.
    However, not every program $\Pi$ is $\psi$-solvable.

When $\psi$-corresponds to $\psi_{\textrm{der}}$, unambiguity comes from
(unicity in)
the Cauchy-Lipschitz theorem.  The fact that
    not every program $\Pi$ is $\psi$-solvable is due to the fact that
    not all differential equations have a solution.



\section{Examples}

The examples in this section are for semantics $\psi_{\textrm{der}}$.
%
Our settings cover, first of all, analog algorithms that are  pure flow, in particular all systems that can be modeled as ordinary differential equations.
A very simple, classical example is the pendulum: the motion of an idealized simple pendulum is governed by the second-order
differential equation $
\theta''+\frac{g}{L}\theta = 0\,,
$
where
$\theta$ is angular displacement, $g$ is gravitational acceleration,
and $L$ is the length of the pendulum rod.
This can indeed be modeled as the program

 \begin{center}
  \begin{minipage}{8cm}
\begin{ttcode}
flow \\
 \s  $\Machine(\theta,\theta_1)$ \\
\s  $\Machine(\theta_1,-\frac{g}{L}\cdot \theta)$ \\
endflow
\end{ttcode}
   \end{minipage}
  \end{center}

\noindent
using the fact that any ordinary differential equation can be put in the form of a vectorial first-order equation, here $\theta_1$ corresponding to the derivative of $\theta$.

As a consequence, our formalism covers very generic classes of continuous-time models of computation, including the GPAC, which corresponds to ordinary differential equations with polynomial right-hand sides \cite{GC03,GBC09}.
Recall that the GPAC was proposed as a mathematical model of
differential analyzers (DAs), one of the most famous analog computer
machines in history.  
Figure~\ref{gpac} (left) depicts a (non-minimal) GPAC that generates sine and cosine.
In this
picture,  $\int$ signifies  some integrator, and  $-1$ denotes
some constant block.
This simple GPAC 
can be modeled by the program

\begin{center}
 \begin{minipage}{8cm}
\begin{ttcode}
flow \\
 \s $\Machine(x,z)$ \\
\s  $\Machine(y,x)$ \\
\s  $\Machine(z,-x)$ \\
endflow
\end{ttcode}
  \end{minipage}
 \end{center}


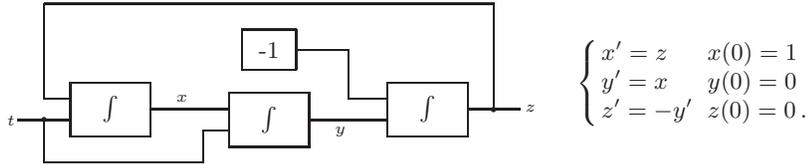
\begin{figure}[t]
\begin{tabular}{ll}
\begin{minipage}{0.5\textwidth}
\begin{center}
\small
\begin{picture}(200,60)(13,10)
\put(40,20){\framebox(30,20){$\int$}}
\put(100,16){\framebox(30,20){$\int$}} \put(100,30){\line(-1,0){30}}
\put(160,20){\framebox(30,20){$\int$}} \put(160,26){\line(-1,0){30}}
\put(105,45){\framebox(20,15){\textrm{-1}}}
\put(160,34){\line(-1,0){15}} \put(145,34){\line(0,1){18.5}}
\put(145,52.5){\line(-1,0){20}} \put(40,26){\line(-1,0){20}}
\put(30,26){\circle*2} \put(30,26){\line(0,-1){16}}
\put(100,22){\line(-1,0){10}} \put(90,22){\line(0,-1){12}}
\put(90,10){\line(-1,0){60}} \put(210,30){\line(-1,0){20}}
\put(200,30){\circle*2} \put(200,30){\line(0,1){40}}
\put(200,70){\line(-1,0){170}} \put(30,70){\line(0,-1){36}}
\put(40,34){\line(-1,0){10}} \put(16,24){$\scriptstyle t$}
\put(212,29){$\scriptstyle z$} \put(140,21){$\scriptstyle
y$} \put(80,33){$\scriptstyle x$}
\end{picture}
\end{center}
\end{minipage}
&
\begin{minipage}{0.5\textwidth}
$$\left\{
\begin{array}[c]{lll}
x^{\prime}=z &  & x(0)=1\\
y^{\prime}=x &  & y(0)=0\\
z^{\prime}=-y^{\prime} &  & z(0)=0\,.
\end{array}
\right.
$$
\end{minipage}
\end{tabular}
\caption{A GPAC for sine and cosine (left). Corresponding evolution (right).}\label{gpac}
\end{figure}



Our proposed model can also adequately describe hybrid
systems, made of alternating sequences of continuous evolution and discrete transitions.
This includes, for example, a simple model of a bouncing ball,
the physics of which are given by the flow equations
$x''  =  - g m $,
where $g$ is the gravitational constant and $v=x'$ is the velocity,
except that upon impact, each time $x=0$, the velocity changes according to
$
v'  =  -k \cdot v',$
where $k$ is the coefficient of impact. Every time the ball
bounces, its speed is reduced by a factor $k$.
This system can be described by a program like

\bigskip

\noindent
\qquad\begin{minipage}{15cm}
\begin{ttcode}
if $x=0$ then \\
\s  $v := -k \cdot v$ \\
else \\
\s flow \\
\s\s   $ \Machine(x,v)$ \\
\s\s $\Machine(v,-g.m)$  \\
\s endflow \\
endif
\end{ttcode}
\end{minipage}

\bigskip

Our setting is an extension of classical discrete-time algorithms; hence, all classical discrete-time algorithms can also be modeled.

 As for examples with semantics other than $ \psi_{\textrm{der}}$:  Observe
  that one can consider timelines like $\mathbb{Q}$ instead of $\R$. (For such
  a timeline, we have $\Continuous{i}$ for all $i \in \mathbb{Q}$.) One
  can define a semantics on such a timeline where for every state $X$
  we have $\Continuous{X}$ by first extending the evolution function
  to $\mathbb{R}$ (for example by
  restricting to continuous
  dynamics) and then using the derivative.
  Constructions of \cite{rust2000hybrid} are also covered by our
  settings: In some sense, the example at the beginning of the paragraph is the
  spirit of the constructions from \cite{rust2000hybrid}, where the
  timeline is the set of hyperreals obtained by multiplying some fixed
  infinitesimal by some hyperinteger (using hyperreals and
  infinitesimals).
Notice  that there is no need to consider derivatives or similar
notions: we could also consider analytic dynamics, and consider a
semantics related to the family of  Taylor coefficients. Weaker
notions of solution, like variational approaches, can also be considered.



\bibliographystyle{splncsa}
\bibliography{biblio2}

\begin{thebibliography}{10}

\bibitem{BSS}
Blum, L., Shub, M., Smale, S.:
\newblock On a theory of computation and complexity over the real numbers; {NP}
  completeness, recursive functions and universal machines.
\newblock Bulletin of the American Mathematical Society \textbf{21} (1989)
  1--46

\bibitem{BD}
Boker, U., Dershowitz, N.:
\newblock {The Church-Turing} thesis over arbitrary domains.
\newblock In: Pillars of Computer Science. Lecture Notes in Computer Science,
  Vol. 4800. Springer (2008)  199--229

\bibitem{3P}
Boker, U., Dershowitz, N.:
\newblock Three paths to effectiveness.
\newblock In: Fields of Logic and Computation.
\newblock Springer (2010)  135--146

\bibitem{Survey}
Bournez, O., Campagnolo, M.L.:
\newblock A survey on continuous time computations.
\newblock In: New Computational Paradigms. Changing Conceptions of What is
  Computable.
\newblock Springer (2008)  383--423

\bibitem{BournezDF12}
Bournez, O., Dershowitz, N., Falkovich, E.:
\newblock Towards an axiomatization of simple analog algorithms.
\newblock In: Proc. 9th Annual Conference on Theory and Applications of Models
  of Computation. Springer (2012)  525--536

\bibitem{report}
{Bournez}, O., {Dershowitz}, N., {N{\'e}ron}, P.:
\newblock {Axiomatizing Analog Algorithms}.
\newblock ArXiv e-prints \url{http://arxiv.org/abs/1604.04295} (2016)

\bibitem{Bush31}
Bush, V.:
\newblock The differential analyser.
\newblock Journal of the Franklin Institute \textbf{212} (1931)  447--488

\bibitem{Cohen}
Cohen, J., Slissenko, A.:
\newblock On implementations of instantaneous actions real-time {ASM} by {ASM}
  with delays.
\newblock In: Proc. 12th Intl. Workshop on Abstract State Machines.
  Universit\'e de Paris 12 (2005)  387--396

\bibitem{Cohen2}
Cohen, J., Slissenko, A.:
\newblock Implementation of sturdy real-time abstract state machines by
  machines with delays.
\newblock In: Proc. 6th Intl. Conf. on Computer Science and Information
  Technology. National Academy of Science of Armenia (2007)

\bibitem{dershowitz2008natural}
Dershowitz, N., Gurevich, Y.:
\newblock A natural axiomatization of computability and proof of {Church's
  Thesis}.
\newblock The Bulletin of Symbolic Logic \textbf{14} (2008)  299--350

\bibitem{fu2015models}
Fu, M.Q., Zucker, J.:
\newblock Models of computation for partial functions on the reals.
\newblock J. Logical and Algebraic Methods in Programming \textbf{84} (2015)
  218--237

\bibitem{GBC09}
Gra{\c{c}}a, D.S., Buescu, J., Campagnolo, M.L.:
\newblock Computational bounds on polynomial differential equations.
\newblock Appl. Math. Comput. \textbf{215} (2009)  1375--1385

\bibitem{GC03}
Gra{\c c}a, D.S., Costa, J.F.:
\newblock Analog computers and recursive functions over the reals.
\newblock Journal of Complexity \textbf{19} (2003)  644--664

\bibitem{Gur00}
Gurevich, Y.:
\newblock Sequential abstract-state machines capture sequential algorithms.
\newblock ACM Trans. Comput. Log. \textbf{1} (2000)  77--111

\bibitem{hasuo2012exercises}
Hasuo, I., Suenaga, K.:
\newblock Exercises in nonstandard static analysis of hybrid systems.
\newblock In: Computer Aided Verification. Springer (2012)  462--478

\bibitem{Nyc96}
Nyce, J.M.:
\newblock Guest editor's introduction.
\newblock IEEE Ann. Hist. Comput. \textbf{18} (1996)  3--4

\bibitem{Platzer08}
Platzer, A.:
\newblock Differential dynamic logic for hybrid systems.
\newblock J. Automated Reasoning \textbf{41} (2008)  143--189

\bibitem{Reisig}
Reisig, W.:
\newblock On {Gurevich's} theorem on sequential algorithms.
\newblock Acta Informatica \textbf{39} (2003)  273--305

\bibitem{rust2000hybrid}
Rust, H.:
\newblock Hybrid abstract state machines: Using the hyperreals for describing
  continuous changes in a discrete notation.
\newblock In: Intl. Workshop on Abstract State Machines. Swiss Federal
  Institute of Technology (2000)  341--356

\bibitem{Sha41}
Shannon, C.E.:
\newblock Mathematical theory of the differential analyser.
\newblock Journal of Mathematics and Physics \textbf{20} (1941)  337--354

\bibitem{suenaga2011programming}
Suenaga, K., Hasuo, I.:
\newblock Programming with infinitesimals: A while-language for hybrid system
  modeling.
\newblock In: Automata, Languages and Programming. Springer (2011)  392--403

\bibitem{TZ}
Tucker, J.V., Zucker, J.I.:
\newblock A network model of analogue computation over metric algebras.
\newblock In: New Computational Paradigms.
\newblock Springer (2005)  515--529

\end{thebibliography}
\end{document}